\documentclass{article}
\usepackage{amsmath}
\usepackage{amssymb}
\usepackage{amsthm}
\usepackage{bm}
\usepackage{apacite}

\usepackage{color}

\newtheorem*{thm}{Theorem}
\newtheorem{prop}{Proposition}
\newtheorem{lem}{Lemma}
\newtheorem{rem}{Remark}

\begin{document}
\title{Interaction Decomposition of a Prediction Function}
\author{Hirokazu Iwasawa, Yoshihiro Matsumori}
\date{}
\maketitle

\begin{abstract}
This paper discusses the foundation of methods for accurately grasping interaction effects. The partial dependence (PD) and accumulated local effects (ALE) methods, which capture interaction effects as terms, are known as global model-agnostic methods in the interpretable machine learning field. ALE provides a functional decomposition of the prediction function. In the present study, we propose and mathematically formalize the requirements of an interaction decomposition (ID) that decomposes a prediction function into its main and interaction effect terms. We also present a theorem by which a decomposition method meeting these requirements can be generated. Furthermore, we confirm that ALE is an ID but PD is not. Finally, we present examples of decomposition methods that meet the requirements of ID, using both existing methods and methods that differ from the existing ones.
\end{abstract}


\section{Introduction}\label{section_introduction}
This paper mathematically discusses the foundation of methods that decompose prediction functions into their main and interaction effect terms. Two existing methods that capture interaction effects as terms are partial dependence \cite{friedman2001greedy} (PD) and accumulated local effects \cite{apley2020visualizing} (ALE). These global model-agnostic methods are well-known in the interpretable machine learning field \cite{adadi2018peeking, lorentzen2020peeking, baeder2021interpretable, molnar2020interpretable} (IML). This paper focuses on an approach that extracts interaction terms through functional decomposition of the prediction function, a theoretical approach used in ALE.

ALE was proposed by Apley et al. \cite{apley2020visualizing}, who pointed out that ALE provides a functional decomposition with excellent decomposition properties (as detailed in Appendix C of their online supplementary materials \cite{apley2016visualizing}). Their Appendix C apparently treats the decomposition properties as unique properties of ALE. However, the present study proposes that these properties are not unique to ALE, but must be satisfied by any functional decomposition method to accurately grasp the interaction effects.

In this paper, we enumerate these properties as well as some additional properties and mathematically formulate specific requirements for functional decomposition methods. We call a functional decomposition method that satisfies the requirements an interaction decomposition (ID). We also provide a theorem by which we can create IDs. Using this theorem, we confirm that ALE is an ID but PD is not. Additionally, we present further examples that meet the requirements of ID, illustrating the versatility of the theorem beyond the existing methods.

The remainder of this paper is structured as follows. Section \ref{section_formulation} mathematically formalizes the ID requirements. Section \ref{section_theorem} presents the theorem and the main proposition of this paper. Section \ref{section_example} uses the theorem to check whether the existing methods meet the ID requirements and provides examples of new methods that meet the ID requirements. Section \ref{section_proof} provides a proof of the main proposition implying the theorem and Section \ref{section_conclusion} concludes the paper.

\section{Formulation}\label{section_formulation}
In this section, we propose and mathematically formalize the requirements that must always be met when decomposing a prediction function into its main and interaction effect terms.
\subsection{Notations and assumptions}\label{subsection_notation&assumptions}
Consider the following functional decomposition of a prediction function $f$:
\[f\left( \mathbf{x}_D\right) =\sum_{J\subseteq D}f_J(\mathbf{x}_J)=f_\emptyset +\sum_{j\in D}f_j(x_j)+\sum_{\{ j,l\} \subseteq D}f_{\{ j,l\} }(x_j,x_l)+\cdots +f_D\left( \mathbf{x}_D\right) .\]

Here we explain the notations of this paper, including the terms in the above formula. This discussion will follow \cite{apley2020visualizing} as far as possible to facilitate comparisons with previous research. In addition, we will frequently itemize our notations for easy reference.
\subsubsection{Basic Notations}\label{subsubsection_notation}
\begin{itemize}
    \item Let $x_1, \ldots, x_d$ be the feature variables. Feature variables treated as random variables are capitalized.
    \item The prediction function $f$ is a function of $d$ variables (in practical applications, $f$ can be replaced with $h \circ f$, where $h$  is a link function.)
    \item The set of variable indices is defined as $D := \{1, \ldots, d\}$.
    \item For a nonempty subset $J \subseteq D$, we define $\mathbf{x}_J := \{x_j | j \in J\}$. In particular, $\mathbf{x}_D := \{x_1, \ldots, x_d\}$. The elements of these sets are variables rather than values. Notations such as $\mathbf{x}_J = \mathbf{c}$ denote that the components of $\mathbf{c}$ are substituted into each variable $x_j$, where $j \in J$.
    \item Define $\setminus J := D \setminus J$ as a set difference; for example,
    \[\mathbf{x}_{\setminus \{j, l\}} = \mathbf{x}_D \setminus \mathbf{x}_{\{j, l\}} = \{x_1, \ldots, x_{j-1}, x_{j+1}, \ldots, x_{l-1}, x_{l+1}, \ldots, x_d\}.\]
    In particular, define $\setminus j := D \setminus \{j\}$. For example, 
    \[\mathbf{x}_{\setminus j} = \mathbf{x}_{\setminus \{j\}} = \{x_1, \ldots, x_{j-1}, x_{j+1}, \ldots, x_d\}.\]
    \item Let $V_J$ represent the space of functions with variables in $\mathbf{x}_J$ and define $V_\emptyset := \mathbb{R}$. If $J' \subseteq J$, then $V_{J'} \subseteq V_J$.
    \item Define $\frac{\partial}{\partial \mathbf{x}_J} := \frac{\partial^{|J|}}{\partial x_{j_1}\cdots\partial x_{j_{|J|}}}$. For example, $\frac{\partial}{\partial \mathbf{x}_{\{j, l\}}} = \frac{\partial^2}{\partial x_j\partial x_l}$.
    \item The function obtained by substituting $\mathbf{c}$ into the variables $\mathbf{x}_J$ of a function $g$ is denoted as $g|_{\mathbf{x}_J=\mathbf{c}}$, where | is the so-called ``evaluated at'' symbol.
    \item Define the partial difference operator ${}_j\Delta_{a_j}: V_D \rightarrow V_D$ as ${}_j\Delta_{a_j}(g) = g - g|_{x_j = a_j}$. For $J = \{j_1, \ldots, j_{|J|}\}$ and $\mathbf{a}_J = (a_{j_1}, \ldots, a_{j_{|J|}})$, define ${}_J\Delta_{\mathbf{a}_J} := {}_{j_1}\Delta_{a_{j_1}} \circ \cdots \circ {}_{j_{|J|}}\Delta_{a_{j_{|J|}}}$. It follows that ${}_J\Delta_{\mathbf{a}_J}(g) = \sum_{J' \subseteq J} (-1)^{|J'|}g|_{\mathbf{x}_{J'} = \mathbf{a}_{J'}}$.
\end{itemize}

\subsubsection{Meaning of the Notation $\frac{\partial g}{\partial \mathbf{x}_J}=0$}\label{subsubsection_notation_partial}
In this paper, the notation $\frac{\partial g}{\partial \mathbf{x}_J}=0$ is extended to cases of discontinuous $x_j$. Specifically, when stating $\frac{\partial g}{\partial \mathbf{x}_J}=0$, we mean that ``${}_J \Delta _{\mathbf{a}_J}\left( g\right) =0$ for any value $\mathbf{a}_J$ that $\mathbf{x}_J$ can take.'' For example, if $\frac{\partial g}{\partial x_j}=0$ (regardless of whether $g$ is differentiable with respect to $x_j$ in the usual sense), then $g-\left( g|_{x_j=a_j}\right) =0$ (identically $0$) for any value $a_j$ that $x_j$ can take. In other words, if $g\in V_D$, then $g\left( x_1,\ldots,x_{j-1},x_j,x_{j+1},\ldots,x_d\right) -g\left( x_1,\ldots,x_{j-1},a_j,x_{j+1},\ldots,x_d\right) =0$. Similarly, if $\frac{\partial^2g}{\partial x_j\partial x_l}=0$ (regardless of whether $g$ is differentiable), then $g-\left( g|_{x_j=a_j}\right) -\left( g|_{x_l=a_l}\right) +\left( g|_{x_j=a_j, x_l=a_l}\right) =0$ for any values $a_j$ and $a_l$ that $x_j$ and $x_l$ can take, respectively.
	
\begin{rem}\label{rem_partial}
From the definition of the partial difference operator, it follows that $\frac{\partial g}{\partial \mathbf{x}_J}=0$ holds if and only if $g$ can be decomposed into a sum of functions belonging to a set $\{g_{J^{\prime}\cup\setminus J} \}_{J^{\prime}\subset J}$, where $g_{J^{\prime}\cup\setminus J}\in V_{J^{\prime}\cup\setminus J}$. In particular, if $\frac{\partial g_J}{\partial \mathbf{x}_J}=0$ for some $g_J\in V_J$, then $g_J$ can be decomposed into a sum of functions belonging to a set $\{g_{J^{\prime}} \}_{J^{\prime}\subset J}$, where $g_{J^{\prime}}\in V_{J^{\prime}}$.
\end{rem}

\subsubsection{Notation of Functional Decomposition}\label{subsubsection_notation_FD}
As mentioned at the beginning of Section \ref{subsection_notation&assumptions}, a prediction function $f$ is decomposed as follows: 
\[f\left( \mathbf{x}_D\right) =\sum_{J\subseteq D}{f_J(\mathbf{x}_J)}=f_\emptyset+\sum_{j\in D}{f_j(x_j)}+\sum_{\{j,l\}\subseteq D}{f_{\{ j,l\} }(x_j,x_l)}+\ldots+f_D\left( \mathbf{x}_D\right),\]
where
\begin{itemize}
    \item $f_\emptyset\in\mathbb{R}$ represents the zero-order effect.
    \item $f_j\left(x_j\right)$ represents the first-order (main) effect of $x_j$.
    \item $f_{\{ j,l\} }(x_j,x_l)$ represents the second-order interaction effect of $\{ x_j,x_l\}$.
    \item $f_J(\mathbf{x}_J)$ represents the $|J|$th-order interaction effect of $\mathbf{x}_J$.
    \item A functional decomposition must specify each of these terms. That is, a functional decomposition gives a set of functions $\left\{f_J\right\}_{J\subseteq D}$ with $2^d$ elements. The set $\left\{f_J\right\}_{J\subseteq D}$ is also referred to as a functional decomposition.
\end{itemize}

We now propose a methodology that specifies a concrete functional decomposition by a set $\left\{ \mathcal{H}_J\right\} _{J\subseteq D}$ of linear operators  $\mathcal{H}_J:V_D\rightarrow V_J$. A functional decomposition method is defined by specifying a set of linear operators $\left\{ \mathcal{H}_J\right\} _{J\subseteq D}$ that defines the functional decomposition as $\left\{ f_J:=\mathcal{H}_J (f)\right\} _{J\subseteq D}$. In specifying $\left\{ \mathcal{H}_J\right\} _{J\subseteq D}$, we adopt the following notations:
\begin{itemize}	
    \item $I:V_D\rightarrow V_D$ is the identity operator.
    \item The operator $\mathbb{E}:V_D\rightarrow\mathbb{R}$ for $g\in V_D$ is defined as $\mathbb{E}\left( g\right):=\mathbb{E}[g(\mathbf{X}_D)]$. (In the main examples discussed later, $\mathcal{H}_{\emptyset}=\mathbb{E}$ always holds and must be satisfied by all IDs).
 \item The operator $\mathbb{E}_{\mathbf{X}_{\setminus J}}:V_D\rightarrow V_J$ for $g\in V_D$ is defined as $\mathbb{E}_{\mathbf{X}_{\setminus J}}\left( g\right) \left( \mathbf{x}_J\right) :=\mathbb{E}[g(\mathbf{x}_J,\mathbf{X}_{\setminus J})]$, where $\mathbb{E}$ on the right-hand side is the usual expectation symbol. (Previewing an example shown later, $\mathrm{PD}_J (\mathbf{x}_J ):=\mathbb{E}[f\left(\mathbf{x}_J,\mathbf{X}_{\setminus J}\right) ]$, so we can write $\mathrm{PD}_J = \mathbb{E}_{\mathbf{X}_{\setminus J}} (f)$).
       \item We denote by $\circ$ a composition operator that combines operators. As mentioned earlier, $V_{J^{\prime}}\subseteq V_J$ whenever $J^{\prime}\subseteq J$. Therefore, $V_J\subseteq V_D$ holds for any $J\subseteq D$. It is important to note that, from these facts, $\circ$ can compose the operators between any elements of $\left\{\mathcal{H}_J\right\}_{J\subseteq D}$.
\end{itemize}

\subsubsection{Assumptions}\label{subsubsection_assumptions}
The assumptions of our study are listed below:
\begin{itemize}
    \item All distributions of $X_1,\ldots,X_d$ have compact support.
    \item For any $J\subseteq D$, any function $g\in V_J$ is bounded.
\end{itemize}

\subsection{Description and Requirements of Interaction Decomposition}\label{subsection_ID requirements}
We propose that Properties (P1) through (P6) below are the requirements of a set of linear operators defining a functional decomposition of a prediction function into its main terms and interaction effect terms. A functional decomposition method that satisfies all of (P1)--(P6) is defined as an interaction decomposition, ID for short, and its requirements are called ID requirements.

\begin{itemize}
\setlength{\leftskip}{15pt}
    \item[(P1)] Unbiasedness: $\mathcal{H}_\emptyset=\mathbb{E}$ and $\mathbb{E}\circ \mathcal{H}_J=0$ for $J\neq\emptyset$.
    \item[(P2)] Relevance: If $\frac{\partial g}{\partial \mathbf{x}_J}=0$, then $\mathcal{H}_J\left( g\right) =0$.
    \item[(P3)] Lean Decomposability: If $g\in V_J$, then $\left( \sum_{J^{\prime}\subseteq J}\mathcal{H}_{J^{\prime}}\right) \left( g\right) =g$.
    \item[(P4)] Idempotence: $\mathcal{H}_J\circ \mathcal{H}_J=\mathcal{H}_J$.
    \item[(P5)] Operational Orthogonality: If $J^{\prime}\neq J$, then $\mathcal{H}_{J^{\prime}}\circ \mathcal{H}_J=0$.
    \item[(P6)] Consistency with PD: When all feature variables are independent, $\left( \sum_{J^{\prime}\subseteq J}\mathcal{H}_{J^{\prime}}\right) =\mathbb{E}_{X_{\setminus J}}$.
\end{itemize}

The Unbiasedness property (P1) means that all terms except for the zero-order effect are unbiased (i.e., zero on average). From a decomposition perspective, it is essential to understand whether the relative effect of each term is positive or negative with respect to the values of the feature variables. To this end, we can naturally set each effect to zero on average. One could argue that this requirement is not necessary but, if a decomposition method satisfies all other requirements, it is always possible to make it additionally satisfy this requirement. Therefore, even if it is not considered as necessary, it can be interpreted as an essential rule that avoids unnecessary arbitrariness. Note that if the decomposition method satisfies (P5), the condition after ``and'' in (P1) is redundant.

The Relevance property (P2) means that a function composed only of invariant (or, in a sense, irrelevant) terms to one or more elements of $\mathbf{x}_J$ has no terms representing the effect of $\mathbf{x}_J$. From a decomposition perspective, parts irrelevant to some elements of $\mathbf{x}_J$ should be represented without including an effect term of $\mathbf{x}_J$. 

The Lean Decomposability property (P3) means that if a prediction function $f$ depends solely on $\mathbf{x}_J$, then $f$ must be sufficiently decomposable only with terms belonging to $\left\{f_{J^{\prime}}\ |\ J^{\prime}\subseteq J\right\}$. From a decomposition perspective, we can naturally require that the original function can be reconstructed by summing all terms ($\left( \sum_{J\subseteq D}\mathcal{H}_J\right) \left(f\right) =f$; see Proposition \ref{prop_decomposability}). However, requiring only the decomposability in this sense is insufficient because even any set $\left\{ \mathcal{H}_J\right\} _{J\subseteq D}$ with irrelevant operators $\mathcal{H}_J$ for $ J\subset D $ satisfies it by only defining $\mathcal{H}_D\left(f\right)$ as $\mathcal{H}_D\left(f\right):=f-(\sum _{J\subset D}\mathcal{H}_J )(f)$. Therefore, we require that $f$ must be decomposable only with indispensable terms to reconstruct $f$.

The Idempotence property (P4) and Operational Orthogonality property (P5) mean that if the effect term $f_J$ of $\mathbf{x}_J$ extracted from $f$ by decomposition is itself decomposed, it remains as $f_J$. In other words, the effect term of $\mathbf{x}_J$ extracted from $f_J$ is exactly $f_J$ (Idempotency), and extracting the effect terms of others than $\mathbf{x}_J$ from $f_J$ gives zero (Operational Orthogonality). These requirements are natural because the terms obtained through decomposition are expected to be indivisible; therefore, they should remain intact after further decomposition attempts. Note that if (P2) is satisfied, then (P5) can be replaced with ``If $J^{\prime}\subset J$, then $\mathcal{H}_{J^{\prime}}\circ \mathcal{H}_J=0$'' because if $J^{\prime}\subseteq J$ does not hold, then $\frac{\partial\mathcal{H}_J\left( g\right)}{\partial \mathbf{x}_{J^{\prime}}}=0$, leading to $\mathcal{H}_{J^{\prime}}\circ \mathcal{H}_J=0$ according to (P2). Moreover, given (P5) and (P3), (P4) is derived as
\[
\mathcal{H}_J\circ \mathcal{H}_J\left( g\right) =\left( \sum_{J^{\prime}\subseteq J}\mathcal{H}_{J^{\prime}}\right) \circ \mathcal{H}_J\left( g\right) = \left( \sum_{J^{\prime}\subseteq J}\mathcal{H}_{J^{\prime}}\right) \left( \mathcal{H}_J\left( g\right)\right) =\mathcal{H}_J\left( g\right). 
\]
In this sense, when (P3) holds, (P4) is mathematically redundant under (P5).

The property of Consistency with PD (P6) means that when all feature variables are independent, the result is consistent with a decomposition based on PD (defined later). This requirement might be somewhat stronger than the other requirements but is necessary for excluding senseless functional decompositions such as
\[
f_J\left( \mathbf{x}_J\right):=f|_{\mathbf{x}_{\setminus J}=0} (\mathbf{x}_J )-\mathbb{E}\left[ f|_{\mathbf{x}_{\setminus J}=0} (\mathbf{X}_J )\right].
\]
Nevertheless, the necessity of such a strong requirement is debatable. In the following discussion, this requirement is treated somewhat separately from the other requirements. 

Note that some of the above-discussed ID requirements have been also identified as ALE properties \cite{apley2020visualizing}. For instance, Apley et al. \cite{apley2020visualizing} postulate that ALE satisfies a property corresponding to (P1) (Unbiasedness). They do not mention properties corresponding to (P2) (Relevance) nor (P3) (Lean Decomposability), but postulate that ALE satisfies what they call the additive recovery property, which is the same as Proposition \ref{prop_centering} in the present paper. They also discuss the importance of (P4) (Idempotence) and (P5) (Operational Orthogonality) in ALE under the name of ``a certain orthogonality-like property.'' They do not mention any property corresponding to (P6) (Consistency with PD).

\subsection{Examples of Necessary Conditions for ID}\label{subsection_necessary_condition}
The following propositions are deduced from the ID requirements and are necessary for an ID:
\begin{prop}\label{prop_decomposability}
$\left( \sum_{J\subseteq D}\mathcal{H}_J\right) \left( f\right) =f.$
\end{prop}
This proposition, which expresses a decomposability property, can be immediately deduced from (P3).

\begin{prop}\label{prop_centering}
When $f(\mathbf{x}_D)=\sum_{j=1}^{d}g_j(x_j)$, then $\mathcal{H}_j\left (f\right) \left( x_j\right) =g_j\left( x_j\right) -\mathbb{E}\left[ g_j\left( X_j\right) \right]$. In particular, when $f\left( \mathbf{x}_D\right) =\beta_0+\sum_{j=1}^{d}\beta_jx_j$ (a first-degree polynomial), then $\mathcal{H}_j\left( f\right) \left( x_j\right) =\beta_jx_j-\beta_j\mathbb{E}\left[ X_j\right]$.
\end{prop}
This proposition is deduced from the linearity of the operators and (P1), (P2), and (P3).

As a specific example, when $f\left( \mathbf{x}_D\right) =x_1+x_2$, then $\mathcal{H}_1\left( f\right) \left( x_1\right) =x_1-\mathbb{E}\left[ X_1\right]$, $\mathcal{H}_2\left( f\right) \left( x_2\right) =x_2-\mathbb{E}\left[ X_2\right]$, and $\mathcal{H}_J\left( f\right) =0$ for $J\notin \left\{ \emptyset,\left\{ 1\right\}, \left\{ 2\right\} \right\}$. As shown later, generalized functional ANOVA does not guarantee this proposition for $f\left( \mathbf{x}_D\right) =x_1+x_2$,  so the method is not an ID.

\begin{prop}
If $\frac{\partial}{\partial \mathbf{x}_J}\left( \mathcal{H}_J\left( g\right) \right) =0$, then $\mathcal{H}_J\left( g\right) =0$.
\end{prop}
This proposition is deduced from (P2) and (P4) and is not satisfied by the PD-based naive decomposition (defined later). 

\section{Theorem}\label{section_theorem}
This section proposes a useful theorem and a proposition implying the theorem. The proposition is proven in Section \ref{section_proof}.

\begin{thm}
If a set of linear operators $\left\{ \mathcal{L}_J:V_D\rightarrow V_J \right\} _{\emptyset \neq J\subseteq D}$ satisfies the following two properties:
\begin{itemize}
\setlength{\leftskip}{15pt}
    \item[(P2)*] If $\frac{\partial g}{\partial \mathbf{x}_J}=0$, then $\frac{\partial \mathcal{L}_J\left( g\right)}{\partial \mathbf{x}_J}=0$,
    \item[(P3)*] If $g\in V_J$, then $\frac{\partial}{\partial \mathbf{x}_J}\left( \mathcal{L}_J\left( g\right) -g\right) =0$,
\end{itemize}
then the set of linear operators $\left\{\mathcal{H}_J:V_D\rightarrow V_J\right\}_{J\subseteq D}$ recursively defined as
\begin{itemize}
\setlength{\leftskip}{15pt}
    \item[(D)] $\mathcal{H}_\emptyset:=\mathbb{E}, \quad\mathcal{H}_J:=\left( I-\sum _{J^{\prime} \subset J}\mathcal{H}_{J^{\prime}} \right) \circ \mathcal{L}_J$
\end{itemize}
yields a functional decomposition method $\left\{\mathcal{H}_J\right\}_{J\subseteq D}$ that satisfies the ID requirements (P1)--(P5).
If $\{ {\mathcal{L}_J}\} _{\emptyset\neq J\subseteq D}$ additionally satisfies the following property:
\begin{itemize}
\setlength{\leftskip}{15pt}
    \item[(P6)* ] When all feature variables are independent, for any $g\in V_D$, 
    \[ \frac{\partial}{\partial \mathbf{x}_J}\left( \mathcal{L}_J\left( g\right) -\mathbb{E}_{X_{\setminus J}}\left( g\right) \right) =0, \]
\end{itemize}
then the above-defined functional decomposition method $\left\{\mathcal{H}_J\right\}_{J\subseteq D}$ satisfies all ID requirements (P1)--(P6).
\end{thm}

\begin{proof}
The proof follows from (1), (2), (5), and (6) of Proposition \ref{prop_main} below.
\end{proof}

\begin{rem}
By the proposed theorem, if a suitable set of linear operators $\{ {\mathcal{L}_J}\} _{\emptyset\neq J\subseteq D}$ is prepared, it is guaranteed that an ID can be constructed based on that. Specific examples are introduced in the next section (Section \ref{section_example}).
\end{rem}

\begin{rem}
Among the properties of the theorem imposed on the set $\{ {\mathcal{L}_J}\} _{\emptyset\neq J\subseteq D}$, (P2)* corresponds to (P2) (Relevance), (P3)* corresponds to (P3) (Lean Decomposability), and (P6)* corresponds to (P6) (Consistency with PD), referring to (3), (4), and (7) of Proposition \ref{prop_main}, respectively.
\end{rem}

\begin{rem}
The set of operators $\left\{\mathcal{H}_J\right\}_{J\subseteq D}$ constructed by (D) can be considered as an adjustment of $\{ {\mathcal{L}_J}\} _{\emptyset\neq J\subseteq D}$ to ensure desirable properties. Indeed, when we define $\mathcal{L}_J:=\mathcal{H}_J$ for $J\ne\emptyset$ using $\left\{\mathcal{H}_J\right\}_{J\subseteq D}$ satisfying (P1)-(P5) ((P1)-(P6)), then it follows that $\{ {\mathcal{L}_J}\} _{\emptyset\neq J\subseteq D}$ satisfies (P2)* and (P3)* ((P2)*, (P3)*, and (P6)*) immediately from (P2) and (P3) ((P2), (P3), and (P6)), and it follows from (P5) that the set of operators constructed by (D) based on $\{ {\mathcal{L}_J}\} _{\emptyset\neq J\subseteq D}$ is $\left\{\mathcal{H}_J\right\}_{J\subseteq D}$ itself.
\end{rem}

\begin{rem}
According to the first part of the theorem (the entire theorem),  satisfying (P2)* and (P3)* ((P2)*, (P3)*, and (P6)*) is a sufficient condition for ensuring that the set $\left\{\mathcal{H}_J\right\}_{J\subseteq D}$ defined by (D) satisfies (P1)--(P5) ((P1)--(P6)). Note that, according to Propositions (3), (4), and (7) of Proposition \ref{prop_main}, they are necessary conditions as well.
\end{rem}

\begin{prop}\label{prop_main}
For the set $\left\{\mathcal{H}_J\right\}_{J\subseteq D}$ defined by (D) in the theorem, 
\begin{enumerate}
\renewcommand{\labelenumi}{(\arabic{enumi}) }
    \item (P1) always holds.
    \item (P2)* and (P3)* $\Rightarrow$ (P2) and (P3).
    \item (P2) $\Rightarrow$ (P2)*.
    \item (P3) $\Rightarrow$ (P3)*.
    \item (P2) and (P3) $\Rightarrow$ (P4) and (P5).
    \item (P2) and (P3) and (P6)* $\Rightarrow$ (P6).
    \item (P6) $\Rightarrow$ (P6)*.
\end{enumerate}
\end{prop}
\begin{rem}
The term (P4) on the right-hand side of (5) can be derived from (P3) and (P5) independently of (D), as mentioned in Section \ref{subsection_ID requirements}; therefore, the right-hand side can be expressed in terms of (P5) alone. However, it is here expressed in terms of (P4) and (P5) because we prove (P4) and then derive (P5) in our later proof.
\end{rem}

\section{Examples}\label{section_example}
Using our theorem, we now verify whether existing methods meet the ID requirements and present new methods that satisfy the ID requirements.

\subsection{Polynomial-based Decomposition}\label{subsection_polynomial-based}
If the prediction function is restricted to a polynomial of degree at most $r\in \mathbb{N}=\left\{ 0,1,2,\ldots \right\}$:
\[f\left( \mathbf{x}_D\right) = \sum _{\substack{0\leq \sum _{j=1}^d r_j \leq r, \\ r_j\in \mathbb{N}}}a_{r_1,\ldots ,r_d} \prod _{j=1}^d x_j^{r_j},\]
then a relatively simple functional decomposition method satisfies ID requirements (P1)--(P5).

As a very simple example, consider the decomposition $\left\{ f_{J,\mathrm{Poly1}}\right\} _{J\subseteq D}$ defined as
\[f_{J,\mathrm{Poly1}}\left( \mathbf{x}_J\right) :=\sum _{\substack{r_j\in \mathbb{N}\setminus \{ 0\} ,\\ \sum _{j\in J}r_j \leq r}}a_{r_1,\ldots ,r_d } \prod _{j\in J}x_j^{r_j}-c_J,\]
where $c_J$ is a constant defined such that $\mathbb{E}\left[ f_{J,\mathrm{Poly1}}\left( \mathbf{X}_J\right) \right] =0$.
	
As another example, consider the decomposition $\left\{ f_{J,\mathrm{Poly2}}\right\} _{J\subseteq D}$  defined as
\[ f_{J,\mathrm{Poly2}}\left( \mathbf{x}_J\right) :=\sum _{\substack{r_j\in \mathbb{N}\setminus \{ 0\} , \\ \sum _{j\in J}r_j \leq r}} b_{r_1,\ldots ,r_d} \prod _{j\in J}\left( x_j-\mathbb{E}[X_j]\right) ^{r_j}-c_J, \]
where $c_J$ is a constant defined such that $\mathbb{E}\left[ f_{J,\mathrm{Poly2}}\left( \mathbf{X}_J\right) \right] =0$ and $b_{r_1,\ldots,r_d}$ for $0\leq \sum _{j=1}^d r_j \leq r$, $r_j\in \mathbb{N}$ are defined such that 
\[f\left( \mathbf{x}_D\right) =\sum _{\substack{0\leq \sum _{j=1}^d r_j \leq r, \\ r_j\in \mathbb{N}}} b_{r_1,\ldots ,r_d} \prod _{j=1}^d \left( x_j-\mathbb{E}[X_j]\right)^{r_j}. \]

From the form of each term, both of the two methods generating these decompositions clearly satisfy ID requirements (P1)--(P5) but their satisfaction of (P6) cannot be generally expected. Therefore, these decomposition methods are not IDs in general.

\subsection{ALE}\label{subsection_ALE}
As shown in \cite{apley2020visualizing}, ALE can be defined for fairly general prediction functions, but here we simplify the definition to investigate the fundamental properties of ALE. This version is limited to sufficiently differentiable prediction functions (the detailed conditions are given in Theorems 1 and 2 of \cite{apley2020visualizing} and the corresponding parts in their online supplementary materials \cite{apley2016visualizing}). Also, in this section (Section \ref{subsection_ALE}) we assume that integrals and partial differentiations with respect to $\mathbf{x}_J$ are exchangeable with $\mathbb{E}_{\mathbf{X}_{\setminus J}}$ and that the minimum value that $X_j$ can take is $x_{\mathrm{min},j}$. In this case, the notation $\frac{\partial g}{\partial \mathbf{x}_J}=0$ is synonymous with the notation of usual partial differentiation and the functional decomposition by ALE $\left\{f_{J,\mathrm{ALE}}\right\}_{J\subseteq D}$ can be recursively defined as follows:
\[\mathcal{L}_{J,\mathrm{ALE}}\left( g\right) \left( \mathbf{x}_J\right) := \int _{\substack{x_{\mathrm{min}, j} < z_j < x_j, \\ j \in J}} \mathbb{E}\left[ \frac{\partial g}{\partial \mathbf{x}_J} \left(\mathbf{X}_D\right) \, \middle| \, \mathbf{X}_J=\mathbf{z}_J \right]  d\mathbf{z}_J, \quad\mathcal{H}_{\emptyset ,\mathrm{ALE}}:=\mathbb{E},\]
\[\mathcal{H}_{J,\mathrm{ALE}}:=\left(I-\sum _{J^{\prime}\subset J}\mathcal{H}_{J^{\prime}, \mathrm{ALE}} \right)\circ \mathcal{L}_{J,\mathrm{ALE}}, \quad f_{J,\mathrm{ALE}}:=\mathcal{H}_{J,\mathrm{ALE}} (f).\]
In particular, we have
\[f_{j,\mathrm{ALE}}\left( x_j\right) =\int _{x_{\mathrm{min},j}}^{x_j}\mathbb{E}\left[ \frac{\partial f}{\partial x_j} \left(\mathbf{X}_D\right) \, \middle| \, X_j=z_j \right] dz_j-c_j,\]
\begin{align*}
f_{\left\{ j,l\right\} ,\mathrm{ALE}}\left( x_j,x_l\right) = & \mathcal{L}_{\{ j,l\} ,\mathrm{ALE}}(f)(x_j,x_l)-\int _{x_{\mathrm{min},j}}^{x_j}\mathbb{E}\left[ \frac{\partial \mathcal{L}_{\{ j,l\} ,\mathrm{ALE}}(f)}{\partial x_j} \left(\mathbf{X}_D\right) \, \middle| \, X_j=z_j \right] dz_j\\
 & -\int _{x_{\mathrm{min},l}}^{x_l}\mathbb{E}\left[ \frac{\partial \mathcal{L}_{\{ j,l\} } ,\mathrm{ALE}(f)}{\partial x_l} \left(\mathbf{X}_D\right)\, \middle| \, X_l=z_l \right] dz_l-c_{jl}, 
\end{align*}
where $c_j$ and $c_{jl}$ are constants defined such that
\[\mathbb{E}\left[ f_{j,\mathrm{ALE}}\left( X_j\right) \right] =0 \text{ and } \mathbb{E}\left[f_{\{j,l\},\mathrm{ALE}}\left( X_j,X_l\right) \right] =0,\] 
respectively, and 
\[\mathcal{L}_{\left\{ j,l\right\} ,\mathrm{ALE}}\left( f\right) \left( x_j,x_l\right) =\int _{x_{\mathrm{min},l}}^{x_l}\int _{x_{\mathrm{min},j}}^{x_j}\mathbb{E}\left[ \frac{\partial ^2 f}{\partial x_j \partial x_l} \left(\mathbf{X}_D\right) \, \middle| \, X_j=z_j,X_l=z_l \right] dz_jdz_l. \]

\begin{rem}
This definition is, as a whole, simpler than the ``Definition for Higher-Order Effects'' in the original paper \cite{apley2020visualizing}. It is also more preferable, among others, in that we simply define $\mathcal{H}_D$ as $\mathcal{H}_J$ in the case $J=D$ rather than enforcing decomposability by separately defining $\mathcal{H}_D$ as in  \cite{apley2020visualizing}.
\end{rem}

The following proposition holds:

\begin{prop}
ALE is an ID.
\end{prop}

\begin{proof}
Because $\left\{\mathcal{L}_{J,\mathrm{ALE}}\right\}_{\emptyset\neq J\subseteq D}$ satisfies the assumptions of theorems (P2)*, (P3)*, and (P6)*, we conclude that ALE satisfies all ID requirements in the theorem. Clearly, $\left\{\mathcal{L}_{J,\mathrm{ALE}}\right\}_{\emptyset\neq J\subseteq D}$ satisfies (P2)* and (P3)*, so here we prove only that it satisfies (P6)*.

When all feature variables are independent, it follows that
\[
\mathbb{E}\left[ \frac{\partial g}{\partial \mathbf{x}_J}\left(\mathbf{X}_D\right)\, \middle| \, \mathbf{X}_J=\mathbf{z}_J \right] =\mathbb{E}_{\mathbf{X}_{\setminus J} } \left( \frac{\partial g}{\partial \mathbf{x}_J}\right) (\mathbf{z}_J),
\]
for any $g\in V_D$.
Therefore,
\begin{align*}
&\frac{\partial}{\partial \mathbf{x}_J}\left( \mathcal{L}_{J,\mathrm{ALE}}\left( g\right) -\mathbb{E}_{X_{\setminus J}}\left( g\right) \right) \left( \mathbf{x}_J\right) \\
& =\frac{\partial }{\partial \mathbf{x}_J}\left( \int _{\substack{x_{\mathrm{min}, j} < z_j < x_j, \\ j \in J}}\mathbb{E}\left[ \frac{\partial g}{\partial \mathbf{x}_J}\left(\mathbf{X}_D\right)\, \middle| \, \mathbf{X}_J=\mathbf{z}_J\right] d\mathbf{z}_J-\mathbb{E}_{X_{\setminus J}}\left( g\right) \left( \mathbf{x}_J\right) \right) \\
& =\frac{\partial }{\partial \mathbf{x}_J}\left( \int _{\substack{x_{\mathrm{min}, j} < z_j < x_j, \\ j \in J}}{\mathbb{E}_{X_{\setminus J}}\left( \frac{\partial g}{\partial \mathbf{x}_J}\right) (\mathbf{z}_J)}d\mathbf{z}_J-\mathbb{E}_{X_{\setminus J}}\left( g\right) \left( \mathbf{x}_J\right) \right) \\
 & =\mathbb{E}_{X_{\setminus J}}\left( \frac{\partial g}{\partial \mathbf{x}_J}\right) \left( \mathbf{x}_J\right) -\frac{\partial }{\partial \mathbf{x}_J}\left( \mathbb{E}_{X_{\setminus J}}\left( g\right) \right) \left( \mathbf{x}_J\right) \\
 & =\mathbb{E}_{X_{\setminus J}}\left( \frac{\partial g}{\partial \mathbf{x}_J}\right) \left( \mathbf{x}_J\right) -\mathbb{E}_{X_{\setminus J}}\left( \frac{\partial g}{\partial \mathbf{x}_J}\right) \left( \mathbf{x}_J\right) =0,
\end{align*}
confirming that $\mathcal{L}_{J,\mathrm{ALE}}$ satisfies (P6)*.
\end{proof}

\begin{rem}
Similarly to the above proposition, the online supplementary materials \cite{apley2016visualizing} for Apley et al. \cite{apley2020visualizing} indicate superior properties of ALE. The main differences between Apley et al.'s and our propositions are discussed in Section \ref{subsection_ID requirements} of this paper.
\end{rem}

\subsection{PD-based Naive Decomposition}\label{subsection_PD-naive}
The one-dimensional and two-dimensional PDs are not intended to provide functional decomposition terms. In particular, the two-dimensional PD captures all effects up to the second-order effect, not merely the additional effect that appears as a second-order term. However, when the significances or strengths of the interaction effects are measured using the PD-based Friedman's H-statistic \cite{friedman2008predictive} or its variant, the unnormalized H-statistic \cite{inglis2022visualizing}, we can naturally expect that the interaction effect terms in the functional decomposition are being extracted.

PD is naturally extendible to three or more dimensions. In fact, in the original paper of Friedman \cite{friedman2001greedy}, PD is defined before partial dependence plots and is not limited to one or two dimensions:
\[
\mathrm{PD}_J\left( \mathbf{x}_J\right) :=\mathbb{E}\left[ f(\mathbf{x}_J,\mathbf{X}_{\setminus J} )\right].
\]
It should be noted that the above notation differs from that in the original paper; in particular, Friedman \cite{friedman2001greedy} introduced partial dependence without using ``PD'' as a notation or even an abbreviation. From the above definition, we can also write $\mathrm{PD}_J = \mathbb{E}_{\mathbf{X}_{\setminus J}}\left( f\right)$, which will be used as required in subsequent discussions.

This definition leads to a recursive definition of the functional decomposition $\{f_{J,\mathrm{PD}}^{*}\} _{J\subseteq D}$, which was not discussed in the original paper or  (to our knowledge) in any subsequent papers. The symbol $*$ has no mathematical meaning but merely distinguishes the symbols of the functions $f_{J,\mathrm{PD}}^{*}$ from $f_{J,\mathrm{PD}}$ used in the PD-based proper decomposition introduced in Section \ref{subsection_PD-proper}.
\[
f_{\emptyset,\mathrm{PD}}^{*}:=\mathbb{E}[f(\mathbf{X}_D )],
\]
\[
f_{J,\mathrm{PD}}^{*}\left( \mathbf{x}_J\right) :=\mathrm{PD}_J (\mathbf{x}_J )-\sum _{\substack{J^{\prime}\subset J,\\ J^{\prime}\neq \emptyset }}f_{J^{\prime},\mathrm{PD}}^{*} (\mathbf{x}_{J^{\prime}})-\mathbb{E}\left[ \mathrm{PD}_J (\mathbf{X}_J )\right].
\]
In particular, we have
\[
f_{j,\mathrm{PD}}^*\left(x_j\right)=\mathrm{PD}_j (x_j )-\mathbb{E}[\mathrm{PD}_j (X_j )],
\]
\begin{align*}
f_{\left\{ j,l\right\} ,\mathrm{PD}}^{*}\left( x_j,x_l\right)= & \mathrm{PD}_{\{ j,l\} } (x_j,x_l )-f_{j,\mathrm{PD}}^{*}(x_j )-f_{l,\mathrm{PD}}^{*}(x_l )-\mathbb{E}[\mathrm{PD}_{\{ j,l\} }(X_j,X_l )]\\
 = & \mathrm{PD}_{\{ j,l\} }(x_j,x_l )-\mathrm{PD}_j (x_j )-\mathrm{PD}_l (x_l )\\
&-\mathbb{E}[\mathrm{PD}_{\{ j,l\} }(X_j,X_l )-\mathrm{PD}_j (X_j )-\mathrm{PD}_l (X_l )].
\end{align*}
Therefore, the H-statistic for $\{j, l\}$ defined in \cite{friedman2008predictive} is essentially based on this decomposition because it can be expressed as 
\[ H_{jl}^2= \frac{\sum_{i=1}^n \hat{f}_{\left\{ j,l\right\}, \mathrm{PD}}^{*}\left( x_{ij},x_{il}\right)^2} {\sum_{i=1}^n\left(\hat{f}_{\left\{ j,l\right\}, \mathrm{PD}}^{*}\left( x_{ij},x_{il}\right) +\hat{f}_{j, \mathrm{PD}}^{*}\left( x_{ij}\right) +\hat{f}_{l, \mathrm{PD}}^{*}\left( x_{il}\right) \right)^2}, \] where $x_{ij}$ and $x_{il}$ are the actual values from the dataset of sample size $n$, and $\hat{f}_{\left\{ j,l\right\}, \mathrm{PD}}^{*}$, $\hat{f}_{j, \mathrm{PD}}^{*}$, and $\hat{f}_{l, \mathrm{PD}}^{*}$ are the estimated versions of $f_{\left\{ j,l\right\}, \mathrm{PD}}^{*}$, $f_{j, \mathrm{PD}}^{*}$, and $f_{l, \mathrm{PD}}^{*}$, respectively.. 

This decomposition can also be defined as $f_{J,\mathrm{PD}}^{*}:=\mathcal{H}_{J,\mathrm{PD}}^{*} (f)$, where the linear operators $\left\{ \mathcal{H}_{J,\mathrm{PD}}^{*}\right\} _{J\subseteq D}$ are defined as
\[\mathcal{H}_{\emptyset ,\mathrm{PD}}^{*}:=\mathbb{E},\quad\mathcal{H}_{J,\mathrm{PD}}^{*}:=(I-\mathbb{E})\circ \left( \mathbb{E}_{\mathbf{X}_{\setminus J}}-\sum _{J^{\prime}\subset J}\mathcal{H}^{*}_{J^{\prime},\mathrm{PD}}\right) .\]

Hereafter, we refer to this functional decomposition method as PD-based naive decomposition, or simply as ``PD-based naive'' when the context is clear.

Despite this decomposition method being implicitly used broadly, from the perspective of our study, it does not possess the desirable properties. Specifically, the following proposition holds:

\begin{prop}
Although PD-based naive satisfies (P1)--(P3), it is not an ID.
\end{prop}

\begin{proof}
This method satisfies (P1) and (P3) by definition. To demonstrate that it also satisfies (P2), it is sufficient to show that
\[\mathcal{H}_{J,\mathrm{PD}}^*(g)(\mathbf{X}_J)
=(I-\mathbb{E})\left(\left.\mathbb{E}\left[(-1)^{|J|}{}_{J}\Delta _{\mathbf{y}_{J}}(g)(\mathbf{X}_D)\right]\right|_{\mathbf{y}_J=\mathbf{X}_J}\right).\]
When $J=\{j\}$, the proof is easily obtained through simple calculations. When $|J|\geq 2$, the following general formula
\[g|_{\mathbf{x}_J=\mathbf{y}_J}(\mathbf{x}_D)
= (I-{}_{j_1}\Delta_{y_{j_1}}) \circ \cdots \circ (I-{}_{j_{|J|}}\Delta_{y_{j_{|J|}}})(g)(\mathbf{x}_D)
= \sum_{J'\subseteq J}(-1)^{|J'|}{}_{J'}\Delta_{\mathbf{y}_{J'}}(g)(\mathbf{x}_D)\]
can be transformed as
\[\left(-1\right)^{|J|}{}_{J}\Delta_{\mathbf{y}_J}(g)(\mathbf{x}_D)
= g|_{\mathbf{x}_J=\mathbf{y}_J}(\mathbf{x}_D) - \sum_{J'\subset J}\left(-1\right)^{|J'|}{}_{J'}\Delta_{\mathbf{y}_{J'}}(g)(\mathbf{x}_D).\]
Replacing $\mathbf{x}_D$ with $\mathbf{X}_D$, the expectation is obtained as
\begin{align*}
\mathbb{E}\left[\left(-1\right)^{|J|}{}_{J}\Delta_{\mathbf{y}_J}(g)(\mathbf{X}_D)\right]
& = \mathbb{E}\left[g|_{\mathbf{x}_J=\mathbf{y}_J}(\mathbf{X}_D)\right] - \sum_{J'\subset J}\mathbb{E}\left[\left(-1\right)^{|J'|}{}_{J'}\Delta_{\mathbf{y}_{J'}}(g)(\mathbf{X}_D)\right] \\
& = \mathbb{E}_{\mathbf{X}_{\setminus J}}(g)(\mathbf{y}_J) - \sum_{J'\subset J}\mathbb{E}\left[\left(-1\right)^{|J'|}{}_{J'}\Delta_{\mathbf{y}_{J'}}(g)(\mathbf{X}_D)\right].
\end{align*}
Setting $\mathbf{y}_J = \mathbf{X}_J$, applying $(I-\mathbb{E})$ to both sides, and noting that $(I-\mathbb{E}) \circ (I-\mathbb{E}) = (I-\mathbb{E})$, we obtain
\begin{align*}
&(I-\mathbb{E})\left(\left.\mathbb{E}\left[\left(-1\right)^{|J|}{}_{J}\Delta_{\mathbf{y}_J}(g)(\mathbf{X}_D)\right]\right|_{\mathbf{y}_J=\mathbf{X}_J}\right) \\
& = (I-\mathbb{E})\left(\mathbb{E}_{\mathbf{X}_{\setminus J}}(g)(\mathbf{X}_J) - \sum_{J'\subset J}(I-\mathbb{E})\left(\left.\mathbb{E}\left[(-1)^{|J'|}{}_{J'}\Delta_{\mathbf{y}_{J'}}(g)(\mathbf{X}_D)\right]\right|_{\mathbf{y}_{J'}=\mathbf{X}_{J'}}\right)\right).
\end{align*}
Thus, $(I-\mathbb{E})\left(\left.\mathbb{E}\left[(-1)^{|J|}{}_{J}\Delta_{\mathbf{y}_J}(g)(\mathbf{X}_D)\right]\right|_{\mathbf{y}_J=\mathbf{X}_J}\right)$ satisfies the same recursive formula as $\mathcal{H}_{J,\mathrm{PD}}^*(g)(\mathbf{X}_J)$.
Therefore, $\mathcal{H}_{J,\mathrm{PD}}^*(g)(\mathbf{X}_J)
= (I-\mathbb{E})\left(\left.\mathbb{E}\left[(-1)^{|J|}{}_{J}\Delta_{\mathbf{y}_J}(g)(\mathbf{X}_D)\right]\right|_{\mathbf{y}_J=\mathbf{X}_J}\right)$ is proven by induction.

By example, we now demonstrate that this method does not satisfy (P4) and (P5) and is therefore not an ID.
Let the prediction function be $f\left( x_1,x_2,x_3\right) =x_1x_2x_3$, and let $X_1=U$, $X_2=U+V$, $X_3=W$, where $U$, $V$, and $W$ are independently distributed with mean 0 and variance 1. In this case, we have
\[
\mathrm{PD}_1\left( x_1\right) =\mathrm{PD}_2\left( x_2\right) =0,
\quad\mathrm{PD}_3\left( x_3\right) =x_3,
\]
\[
\mathrm{PD}_{\{ 1,2\} }\left( x_1,x_2\right) =\mathrm{PD}_{\{1,3\}}\left( x_1,x_3\right) =\mathrm{PD}_{\{ 2,3\} }\left( x_2,x_3\right) =0,
\]
\[
f_{\{ 1,2\},\mathrm{PD}}^{*}\left( x_1,x_2\right) =0,
\quad f_{\{ 1,3\},\mathrm{PD}}^{*}\left( x_1,x_3\right) =f_{\{ 2,3\},\mathrm{PD}}^{*}\left( x_2,x_3\right) = -x_3.
\]
Note that $f_{\{ 1,3\}, \mathrm{PD}}^{*}\left( x_1,x_3\right)$ (for example) is essentially a function of one variable $(-x_3)$. Therefore, applying $\mathcal{H}_{\{ 1,3\} ,\mathrm{PD}}^{*}$ to this function gives 0, which violates (P4). In contrast, applying $\mathcal{H}_{3,\mathrm{PD}}^{*}$ to this function gives a nonzero value, violating (P5). The same analysis on $f_{\{ 2,3\} ,\mathrm{PD}}^{*}\left( x_2,x_3\right)$ obtains the same results. Therefore, PD-based naive is not an ID.
\end{proof}

\begin{rem}
In the example in the above proof, $f_{\{ 1,3\}, \mathrm{PD}}^{*}\left( x_1,x_3\right)$ (for example) is essentially a function of one variable. However, an interaction effect term is not expected to be a function of one variable from a common-sense perspective. Therefore,  besides deviating from the ID requirements defined in this paper, PD-based naive is inappropriate from a practical viewpoint. However, PD-based naive is not entirely unreasonable because it satisfies (P2).
\end{rem}

\subsection{PD-based Proper Decomposition}\label{subsection_PD-proper}

As analyzed above, the ``natural'' PD-based functional decomposition method implicitly assumed in Friedman's H-statistic does not satisfy the ID requirements. Nevertheless, we can construct an ID based on PD.

Concretely, we can recursively define a functional decomposition $\left\{ f_{J,\mathrm{PD}}\right\} _{J\subseteq D}$ as follows:
\[\mathcal{L}_{J,\mathrm{PD}}:=\mathbb{E}_{\mathbf{X}_{\setminus J}}, \quad \mathcal{H}_{\emptyset,\mathrm{PD}}:=\mathbb{E},\]
\[\mathcal{H}_{J,\mathrm{PD}}:=\left( I-\sum_{J^{\prime}\subset J}\mathcal{H}_{J^{\prime},\mathrm{PD}} \right) \circ \mathcal{L}_{J,\mathrm{PD}}, \quad f_{J,\mathrm{PD}}:=\mathcal{H}_{J,\mathrm{PD}}(f).\]
In particular, we have
\[f_{j,\mathrm{PD}}\left(x_j\right)=\mathrm{PD}_j (x_j )-\mathbb{E}[\mathrm{PD}_j (X_j )],\]
\[f_{\{ j,l\} ,\mathrm{PD}}\left(x_j,x_l\right)=\mathrm{PD}_{\{ j,l\} }(x_j,x_l )-\mathbb{E}[\mathrm{PD}_{\{ j,l\} } (x_j,X_l )]-\mathbb{E}[\mathrm{PD}_{\{ j,l\} } (X_j,x_l )]-c_{jl},\]
where $c_{jl}$ is a constant defined such that $\mathbb{E}\left[ f_{\{ j,l\} ,\mathrm{PD}}\left( X_j,X_l\right) \right] =0$. 

We refer to the method generating this functional decomposition as PD-based proper decomposition, or simply ``PD-based proper'' when the context is clear.

The following propositions hold:

\begin{prop}
PD-based proper is an ID.
\end{prop}

\begin{proof}
If $\frac{\partial g}{\partial \mathbf{x}_J}=0$, then 
\[ {}_J\Delta _{\mathbf{a}_J}\mathcal{L}_{J,\mathrm{PD}}\left( g\right) \left( \mathbf{x}_J\right) ={}_J\Delta _{\mathbf{a}_J}\mathbb{E}\left[ g\left( \mathbf{x}_J,\mathbf{X}_{\setminus J}\right) \right] =\mathbb{E}\left[ {}_J\Delta _{\mathbf{a}_J}g\left( \mathbf{x}_J,\mathbf{X}_{\setminus J}\right) \right] =0. \] 
Therefore, $\left\{ \mathcal{L}_{J,\mathrm{PD}}\right\} _{\emptyset \neq J\subseteq D}$ satisfies (P2)*. For any $g\in V_J$, we have $\mathcal{L}_{J,\mathrm{PD}}\left( g\right) \left( \mathbf{x}_J\right) =\mathbb{E}\left[g\left( \mathbf{x}_J\right) \right]=g\left( \mathbf{x}_J\right)$, meaning that $\left\{\mathcal{L}_{J,\mathrm{PD}}\right\}_{\emptyset\neq J\subseteq D}$ satisfies (P3)*. Furthermore, as $\mathcal{L}_{J,\mathrm{PD}}=\mathbb{E}_{X_{\setminus J}}$, $\left\{\mathcal{L}_{J,\mathrm{PD}}\right\}_{\emptyset\neq J\subseteq D}$ also satisfies (P6)* (see Proposition \ref{prop_consistency_of_PDs}). Therefore, PD-based proper is an ID under the proposed theorem.
\end{proof}

\begin{prop}\label{prop_consistency_of_PDs}
If all feature variables are independent, then PD-based proper and PD-based naive are equivalent.
\end{prop}

\begin{proof}
Suppose that all feature variables are independent. Note that in this case, $\mathbb{E}_{X_{\setminus J'}}\circ \mathbb{E}_{X_{\setminus J}}=\mathbb{E}_{X_{\setminus J'}}$ for any $J' \subset J$. On the one hand, because PD-based naive satisfies (P1), we have, for nonempty $J\subseteq D$,
\begin{align*}
\mathcal{H}_{J,\mathrm{PD}}^{*}&=\left( I-\mathbb{E}\right) \circ \left( \mathbb{E}_{\mathbf{X}_{\setminus J}}-\sum_{J'\subset J}\mathcal{H}_{J',\mathrm{PD}}^{*}\right)\\ &=\mathbb{E}_{\mathbf{X}_{\setminus J}}-\mathbb{E}-\sum_{J'\subset J}\mathcal{H}_{J',\mathrm{PD}}^{*}+\mathbb{E}
=\mathbb{E}_{\mathbf{X}_{\setminus J}}-\sum_{J'\subset J}\mathcal{H}_{J',\mathrm{PD}}^{*}.\end{align*}
Therefore, $\sum_{J'\subseteq J}\mathcal{H}_{J',\mathrm{PD}}^{*}=\mathbb{E}_{X_{\setminus J}}$. On the other hand, for PD-based proper with nonempty $J' \subset J$, we have
\begin{align*}
\mathcal{H}_{J',\mathrm{PD}}\circ \mathcal{L}_{J,\mathrm{PD}}
&=\left( I-\sum_{J''\subset J'}\mathcal{H}_{J'',\mathrm{PD}}\right) \circ \mathcal{L}_{J',\mathrm{PD}}\circ \mathcal{L}_{J,\mathrm{PD}}\\
&=\left( I-\sum_{J''\subset J'}\mathcal{H}_{J'',\mathrm{PD}}\right) \circ \mathcal{L}_{J',\mathrm{PD}}
=\mathcal{H}_{J',\mathrm{PD}},
\end{align*}
and $\mathcal{H}_{\emptyset,\mathrm{PD}}\circ \mathcal{L}_{J,\mathrm{PD}}=\mathcal{H}_{\emptyset,\mathrm{PD}}$. Thus,
\[\mathcal{H}_{J,\mathrm{PD}}=\left( I-\sum_{J'\subset J}\mathcal{H}_{J',\mathrm{PD}}\right) \circ \mathcal{L}_{J,\mathrm{PD}}=\mathbb{E}_{\mathbf{X}_{\setminus J}}-\sum_{J'\subset J}\mathcal{H}_{J',\mathrm{PD}},\]
meaning that $\sum_{J'\subseteq J}\mathcal{H}_{J',\mathrm{PD}}=\mathbb{E}_{X_{\setminus J}}$ and hence $\sum_{J'\subseteq J}\mathcal{H}_{J',\mathrm{PD}}^{*}=\sum_{J'\subseteq J}\mathcal{H}_{J',\mathrm{PD}}=\mathbb{E}_{X_{\setminus J}}$. Therefore, by induction, PD-based proper and PD-based naive are equivalent.
\end{proof}

\subsection{Functional ANOVA}\label{subsection_functionalANOVA}

Functional ANOVA has long been discussed and developed in various forms. To relate functional ANOVA to ALE, we can naturally refer to the formulation of \cite{hooker2007generalized} or its prior work \cite{hooker2004discovering} (see \cite{apley2020visualizing} and \cite{molnar2020interpretable}). However, these papers adopt the formulation of \cite{owen2003dimension}, which assumes that $\mathbf{x}_D$ takes values in $[0,1]^d$ and which does not align with the context of this paper.
Thus, we adopt the formulation of \cite{efron1981jackknife} (see \cite{roosen1995visualization} for details), which is also referred to in \cite{apley2020visualizing}, altering its notation to more closely align with the present paper.

Specifically, the functional decomposition by functional ANOVA $\left\{ f_{J , \mathrm{FA}}\right\} _{J\subseteq D}$ can be recursively defined as follows:
\[f_{\emptyset , \mathrm{FA}}:=\mathbb{E}[f(\mathbf{X}_D)],\]
\begin{align*}
f_{J , \mathrm{FA}}\left( \mathbf{x}_J\right)
& :=\mathbb{E}\left[ f(\mathbf{X}_D)-\sum_{J^{\prime}\subset J}f_{J^{\prime} , \mathrm{FA}}(\mathbf{X}_{J^{\prime}})\, \middle| \, \mathbf{X}_J=\mathbf{x}_J \right]\\
& =\mathbb{E}[f(\mathbf{X}_D) | \mathbf{X}_J=\mathbf{x}_J]-\sum_{J^{\prime}\subset J}f_{J^{\prime} , \mathrm{FA}}(\mathbf{x}_{J^{\prime}}).
\end{align*}
In particular, we have
\[f_{j , \mathrm{FA}}\left( x_j\right) =\mathbb{E}\left[ f\left( \mathbf{X}_D\right) |X_j=x_j\right] -f_{\emptyset , \mathrm{FA}},\]
\[f_{\{ j,l\} , \mathrm{FA} }\left( \mathbf{x}_{\{ j,l\} }\right) =\mathbb{E}\left[ f\left( \mathbf{X}_D\right) |\mathbf{X}_{\left\{ j,l\right\} }=\mathbf{x}_{\left\{ j,l\right\} }\right] -f_{j , \mathrm{FA}}\left( x_j\right) -f_{l , \mathrm{FA}}\left( x_l\right) +f_{\emptyset , \mathrm{FA}}.\]

The formulation in the original context assumes that all feature variables are independent, meaning that $\mathbb{E}\left[ f_{J , \mathrm{FA}}\left( \mathbf{X}_J\right) f_{J^{\prime} , \mathrm{FA}}\left( \mathbf{X}_{J^{\prime}}\right) \right] =0$ when $J\neq J^{\prime}$. This type of orthogonality, for which $\mathbb{E}\left[ f_{J , \mathrm{FA}}\left( \mathbf{X}_J\right) f_{J^{\prime} , \mathrm{FA}}\left( \mathbf{X}_{J^{\prime}}\right) \right]$  can be considered as an inner product, is essential for functional ANOVA. Under the orthogonality condition, this method dissociates the variance, as suggested by its name. These points are not further discussed as they are outside the scope of this paper.

When the feature variables are not independent, the decomposition results are not orthogonal and the variance is not neatly dissociated. In such a case, even Proposition \ref{prop_centering} is not guaranteed. For example, if $f(x_1,x_2)=x_1+x_2$ where $X_1\sim U$, $X_2\sim U+V$, and $U$ and $V$ are independently standard normally distributed, we have $f_{1, \mathrm{FA}}(x_1)=2x_1 \ne x_1 = x_1 -\mathbb{E}\left[X_1\right] $ and the decomposition method is clearly not an ID.

The results in cases of independent feature variables theoretically coincide with PD-based decompositions. However, in practice, as the true function is not explicitly given, the obtained results can differ because functional ANOVA and PD in the literature use different estimation methods for decomposition based on the data.

The generalized functional ANOVA proposed in \cite{hooker2007generalized} provides a decomposition without assuming independence among feature variables. As this approach builds upon the discussions in \cite{hooker2004discovering}, it effectively assumes that feature variables follow uniform distributions as a starting point and generalizes the decomposition with weighting. Therefore, it cannot be simply compared with our present formulation and is not specifically defined here. Instead, we summarize that the method provides a decomposition through optimization under the constraint of specific orthogonality (if $J^{\prime}\subset J$, the inner product of $f_J\left( \mathbf{x}_J\right)$ and $f_{J^{\prime}}\left( \mathbf{x}_{J^{\prime}}\right)$ is 0) when the distribution of feature variables is unknown. The decomposition method, similarly to its predecessor functional ANOVA, is not an ID because it also does not guarantee Proposition \ref{prop_centering} (see Section 5.6 in \cite{apley2020visualizing}).

\subsection{CE-based Decomposition}\label{subsection_CE-based}
Utilizing our theorem, we can devise previously unreported IDs simply by finding a set $\left\{ \mathcal{L}_J\right\} _{\emptyset \neq J\subseteq D}$ that satisfies the assumptions of the theorem ((P2)*, (P3)*, and (P6)*). An example is shown below.

We first select a suitable representative value $\mathrm{Rep}_j$ from the possible values of each feature variable $x_j$. Typical representative values are the expectation, median, and mode. Denoting the vector of representative values $\left( \mathrm{Rep}_j \right) _{j\in J}$ as $\mathrm{Rep}_J$, we define
\[\mathrm{CE}_J\left( \mathbf{x}_J\right) :=\mathbb{E}\left[ {}_J\Delta _{\mathrm{Rep}_J}(f)\left( \mathbf{x}_J,\mathbf{X}_{\setminus J}\right) \, \middle| \, \mathbf{X}_J=\mathbf{x}_J \right].\]
In particular, we have
\[\mathrm{CE}_j\left( x_j\right) =\mathbb{E}\left[ f(x_j,\mathbf{X}_{\setminus j})-f(\mathrm{Rep}_j, \mathbf{X}_{\setminus j}) \, \middle| \, X_j=x_j \right].\]

Intuitively, we note that $\mathrm{CE}_J\left( \mathbf{x}_J\right)$ uses the difference from the reference point when $\mathbf{x}_J=\mathrm{Rep}_J$ and measures the average effect at each value of $\mathbf{x}_J$ using the conditional expectation.

Based on $\mathrm{CE}_J$, we recursively define the functional decomposition $\left\{ f_{J,\mathrm{CE}}\right\} _{J\subseteq D}$ similarly to PD-based proper as follows:
\[\mathcal{L}_{J,\mathrm{CE}}\left( g\right) \left( \mathbf{x}_J\right) :=\mathbb{E}\left[ {}_J\Delta _{\mathrm{Rep}_J}(g)(\mathbf{x}_J,\mathbf{X}_{\setminus J}) \, \middle| \, \mathbf{X}_J=\mathbf{x}_J \right], 
\quad\mathcal{H}_{\emptyset ,\mathrm{CE}}:=\mathbb{E},\]
\[\mathcal{H}_{J,\mathrm{CE}}:=\left( I-\sum _{J^{\prime}\subset J}\mathcal{H}_{J^{\prime},\mathrm{CE}}\right) \circ \mathcal{L}_{J,\mathrm{CE}}, 
\quad f_{J,\mathrm{CE}}:=\mathcal{H}_{J,\mathrm{CE}}(f).\]
In particular, we have
\[f_{j,\mathrm{CE}}\left( x_j\right) = \mathrm{CE}_j (x_j )-\mathbb{E}[\mathrm{CE}_j (X_j )],\]
\begin{align*}
&f_{\left\{ j,l\right\} ,\mathrm{CE}}\left( x_j,x_l\right)\\
& = \mathrm{CE}_{\{ j,l\} }(x_j,x_l )-\mathbb{E}\left[ \mathrm{CE}_{\{ j,l\} }(x_j,X_l )-\mathrm{CE}_{\{ j,l\} }(\mathrm{Rep}_j, X_l ) \, \middle| \, X_j=x_j \right]\\
&\quad -\mathbb{E}\left[ \mathrm{CE}_{\{ j,l\} }(X_j,x_l )-\mathrm{CE}_{\{ j,l\} } (X_j,\mathrm{Rep}_l) \, \middle| \, X_l=x_l \right] -c_{jl},
\end{align*}
where $c_{jl}$ is a constant defined such that $\mathbb{E}\left[ f_{\left\{ j,l\right\} ,\mathrm{CE}}\left( X_j, X_l\right) \right] =0$. 

We refer to the method generating this functional decomposition as CE-based decomposition, or simply ``CE-based'' when the context is clear.

The following proposition holds:

\begin{prop}
CE-based is an ID.
\end{prop}

\begin{proof}
$\left\{ \mathcal{L}_{J,\mathrm{CE}}\right\} _{\emptyset \neq J\subseteq D}$ satisfies all assumptions of the theorem, namely, (P2)*, (P3)*, and (P6)*.
In fact, if $\frac{\partial g}{\partial \mathbf{x}_J}=0$, then ${}_J\Delta _{\mathrm{Rep}_J }\left( g\right) =0$, so 
\[\mathcal{L}_{J,\mathrm{CE}}\left( g\right) \left( \mathbf{x}_J\right) =\mathbb{E}\left[ {}_J\Delta _{\mathrm{Rep}_J }\left( g\right) \left( \mathbf{x}_J,\mathbf{X}_{\setminus J}\right) \, \middle| \, \mathbf{X}_J=\mathbf{x}_J\right] =0,\] 
satisfying (P2)*.
If $g\in V_J$, then
\begin{align*}
\mathcal{L}_{J,\mathrm{CE}}\left( g\right) \left( \mathbf{x}_J\right)
& =\mathbb{E}\left[ {}_J\Delta _{\mathrm{Rep}_J }\left( g\right) \left( \mathbf{x}_J,\mathbf{X}_{\setminus J}\right) \, \middle| \, \mathbf{X}_J =\mathbf{x}_J\right]\\
& ={}_J\Delta _{\mathrm{Rep}_J }\left( g\right) \left( \mathbf{x}_D\right)
 =g\left( \mathbf{x}_D\right) + \mathrm{other\, terms}
\end{align*}
holds. If $\frac{\partial}{\partial \mathbf{x}_J}$ is applied to other terms, the result is $0$ and (P3)* holds.
When the feature variables are independent, we have
\begin{align*}
\mathcal{L}_{J,\mathrm{CE}}\left( g\right) \left( \mathbf{x}_J\right) & =\mathbb{E}\left[ {}_J\Delta _{\mathrm{Rep}_J }\left( g\right) \left( \mathbf{x}_J,\mathbf{X}_{\setminus J}\right) \, \middle| \, \mathbf{X}_J=\mathbf{x}_J\right] \\
& =\mathbb{E}\left[ {}_J\Delta _{\mathrm{Rep}_J }\left( g\right) \left( \mathbf{x}_J,\mathbf{X}_{\setminus J}\right) \right] \\
& =\mathbb{E}\left[ g\left( \mathbf{x}_J,\mathbf{X}_{\setminus J}\right) \right] + \mathrm{other\, terms}.
\end{align*}
Because $\frac{\partial}{\partial \mathbf{x}_J}(\mathrm{other\, terms})=0$, 
$\left\{ \mathcal{L}_{J,\mathrm{CE}}\right\} _{\emptyset \neq J\subseteq D}$ satisfies (P6)*.
Therefore, under the proposed theorem, CE-based is an ID.
\end{proof}

\subsection{RP-based Decomposition}\label{subsection_RP-based}
Using the proposed theorem, we can devise a functional decomposition method that fulfils most but not all ID requirements; specifically, a method that satisfies (P1)--(P5). As mentioned in \ref{subsection_ID requirements}, ID requirement (P6) might be too stringent and alternatives that violate this requirement should be explored. To this end, we can simply find a set $\left\{ \mathcal{L}_J\right\} _{\emptyset \neq J\subseteq D}$ that satisfies (P2)* and (P3)* of the theorem. An example is given below.

Using $\mathrm{Rep}_{\setminus J}$ introduced in \ref{subsection_CE-based}, we first define 
\[ \mathrm{RP}_J\left( \mathbf{x}_J\right) :=f(\mathbf{x}_J,\mathrm{Rep}_{\setminus J})=f|_{\mathbf{x}_{\setminus J}=\mathrm{Rep}_{\setminus J}} (\mathbf{x}_J ). \]
For an intuitive description, we note that whereas $\mathrm{PD}_J\left( \mathbf{x}_J\right)$ takes the average of $f\left( \mathbf{x}_J,\mathbf{X}_{\setminus J}\right)$, the $\mathrm{RP}_J\left( \mathbf{x}_J\right)$ simply takes a representative value of $f\left( \mathbf{x}_J,\mathbf{X}_{\setminus J}\right)$. It can reveal aspects that could be lost through averaging.

\begin{rem}
The $\mathrm{RP}_J\left( \mathbf{x}_J\right)$ might be more clearly understood as a representative of Individual Conditional Expectation (ICE) \cite{goldstein2015peeking}. However, as instances in which all feature variable values match the representative values are not guaranteed, the term `representative' is not necessarily appropriate in the context of ICE.
\end{rem}

Based on $\mathrm{RP}_J\left( \mathbf{x}_J\right)$, we recursively define a functional decomposition $\left\{ f_{J,\mathrm{RP}}\right\} _{J\subseteq D}$ similarly to PD-based proper as follows:
\[\mathcal{L}_{J,\mathrm{RP}}\left( g\right) := g|_{\mathbf{x}_{\setminus J}=\mathrm{Rep}_{\setminus J}},
\quad\mathcal{H}_{\emptyset ,\mathrm{RP}}:=\mathbb{E},\]
\[\mathcal{H}_{J,\mathrm{RP}}:=\left( I-\sum _{J^{\prime}\subset J}\mathcal{H}_{J^{\prime},\mathrm{RP}} \right) \circ \mathcal{L}_{J,\mathrm{RP}},
\quad f_{J,\mathrm{RP}}:=\mathcal{H}_{J,\mathrm{RP}} (f).\]
In particular, we have
\[f_{j,\mathrm{RP}}\left( x_j\right) =\mathrm{RP}_j (x_j )-\mathbb{E}[\mathrm{RP}_j (X_j )],\]
\begin{align*}
f_{\left\{ j,l\right\} ,\mathrm{RP}}\left( x_j,x_l\right) =\mathrm{RP}_{\{ j,l\} } (x_j,x_l )-\mathrm{RP}_j (x_j )-\mathrm{RP}_l (x_l )-c_{jl},
\end{align*}
where $c_{jl}$ is a constant defined such that $\mathbb{E}\left[ f_{\left\{ j,l\right\} ,\mathrm{RP}}\left( X_j, X_l\right) \right] =0$. 

We refer to the method generating this functional decomposition as the RP-based decomposition, or simply ``RP-based'' when the context is clear.

The following proposition holds:

\begin{prop}
RP-based satisfies ID requirements (P1)--(P5).
\end{prop}

\begin{proof}  
By definition, RP-based clearly satisfies assumptions (P2)* and (P3)* of the theorem. It follows from the theorem that RP-based also fulfils ID requirements (P1)--(P5).
\end{proof}

\subsection{Hybrid Decompositions}\label{subsection_hybrid}
When constructing an ID or a decomposition method similar to an ID (such as RP-based) using the proposed theorem, the elements of $\left\{ \mathcal{L}_J\right\} _{\emptyset \neq J\subseteq D}$ can be of different kinds. For example, if ALE exhibits superior aspects for one-dimensional terms and RP-based has superior aspects for calculations in two or more dimensions, the two types of elements can be conceivably combined.

Under the same assumptions as ALE in Section \ref{subsection_ALE}, a hybrid functional decomposition $\left\{ f_{J,\mathrm{ALE}+\mathrm{RP}}\right\} _{J\subseteq D}$ can be recursively defined as follows:
\[\mathcal{H}_{\emptyset,\mathrm{ALE}+\mathrm{RP}}:=\mathbb{E}, 
\quad\mathcal{H}_{j,\mathrm{ALE}+\mathrm{RP}}:=\mathcal{H}_{j,\mathrm{ALE}}.\]
When $\left| J\right| \in \left\{ 2,\ldots,d\right\}$, we have
\[\mathcal{L}_{J,\mathrm{ALE}+\mathrm{RP}}(g):= g|_{\mathbf{x}_{\setminus J}=\mathrm{Rep}_{\setminus J}},\]
\[\mathcal{H}_{J,\mathrm{ALE}+\mathrm{RP}}:=\left( I-\sum _{J^{\prime}\subset J}\mathcal{H}_{J^{\prime},\mathrm{ALE}+\mathrm{RP}}\right) \circ \mathcal{L}_{J,\mathrm{ALE}+\mathrm{RP}},\]
\[ f_{J,\mathrm{ALE}+\mathrm{RP}}:=\mathcal{H}_{J,\mathrm{ALE}+\mathrm{RP}} (f).\]
In particular, we have
\[f_{j,\mathrm{ALE}+\mathrm{RP}}\left( x_j\right) =f_{j,\mathrm{ALE}} (x_j )=\int _{x_{\mathrm{min},j}}^{x_j}\mathbb{E}\left[ {\partial f}{\partial x_j} \, \middle| \, X_j=z_j \right] dz_j-c_j,\]
\begin{align*}
f_{\left\{ j,l\right\} ,\mathrm{ALE}+\mathrm{RP}}\left( x_j,x_l\right) =
&\mathrm{RP}_{\{ j,l\} } (x_j,x_l )-\int _{x_{\mathrm{min},j}}^{x_j}\mathbb{E}\left[ \frac{\partial \mathrm{RP}_{\{ j,l\} }}{\partial x_j} \, \middle| \, X_j=z_j \right] dz_j\\
&-\int _{x_{\mathrm{min},l}}^{x_l}\mathbb{E}\left[ \frac{\partial \mathrm{RP}_{\{ j,l\} }}{\partial x_l} \, \middle| \, X_l=z_l \right] dz_l-c_{jl},
\end{align*}
where $c_j$ and $c_{jl}$ are constants defined such that 
\[ \mathbb{E}\left[ f_{j,\mathrm{ALE}+\mathrm{RP}}\left( X_j\right) \right] =0 \text{ and } \mathbb{E}\left[ f_{\left\{ j,l\right\} ,\mathrm{ALE}+\mathrm{RP}}\left( X_j,X_l\right) \right] =0,\] 
respectively.

The following proposition holds:
\begin{prop}
The above-defined hybrid decomposition method satisfies ID requirements (P1)--(P5).
\end{prop}

\begin{proof}   
By definition, the method clearly satisfies assumptions (P2)* and (P3)* of the theorem. From the theorem, it follows that the method also fulfils ID requirements (P1)--(P5).
\end{proof}

\section{Proof}\label{section_proof}
This section provides a proof of Proposition \ref{prop_main}, which includes items (1) through (7).

To prove this proposition, we define $\left\{ \mathcal{H}_J\right\} _{J\subseteq D}$ as (D) and introduce linear operators and a lemma. These preliminary definitions will be repeatedly used in the following discussion.

We first define the following set of linear operators $\left\{ \mathcal{K}_J:V_D\rightarrow V_J\right\} _{\emptyset \neq J\subseteq D}$:
\[
\mathcal{K}_j:=\mathcal{L}_j, \quad \mathcal{K}_J:=\left( I-\sum _{\emptyset \neq J^{\prime}\subset J}\mathcal{K}_{J^{\prime}} \right) \circ \mathcal{L}_J.
\]
By induction, we can show that $\mathcal{H}_J=\left( I-\mathbb{E}\right) \circ \mathcal{K}_J$ for all $J\neq\emptyset$.

\begin{lem}
For $\emptyset \neq J \subseteq D$, suppose that the following holds for any $J^{\prime} \subset J$: If $\frac{\partial g}{\partial \mathbf{x}_{J^{\prime}}}=0$, then $\mathcal{H}_{J^{\prime}}\left( g\right) =0$ (corresponding to (P2)), and if $g\in V_{J^{\prime}}$, then $\left( \sum_{J^{\prime \prime}\subseteq J^{\prime}}\mathcal{H}_{J^{\prime \prime}}\right) \left( g\right) =g$ (corresponding to (P3)). Then, for any $h\in V_J$ with $\frac{\partial h}{\partial \mathbf{x}_J}=0$, it follows that $\left( \sum_{J^{\prime}\subset J}\mathcal{H}_{J^{\prime}}\right) \left( h\right) =h$.
\end{lem}

\begin{proof}
As $h\in V_J$ and $\frac{\partial h}{\partial \mathbf{x}_J}=0$, it follows from the definition of the partial differentiation symbol that $h$ can be decomposed as $h=\sum_{J^{\prime \prime}\subset J}h_{J^{\prime \prime}}$ using some $h_{J^{\prime \prime}}\in V_{J^{\prime \prime}}$ (see Remark \ref{rem_partial}). As $h_{J^{\prime \prime}}\in V_{J^{\prime \prime}}$, $\left( \sum_{J^{\prime}\subseteq J^{\prime \prime}}\mathcal{H}_{J^{\prime}}\right) \left( h_{J^{\prime \prime}}\right) =h_{J^{\prime \prime}}$ holds. In addition, if $J^{\prime}\subseteq J^{\prime \prime}$ does not hold, then $\frac{\partial h_{J^{\prime \prime}}}{\partial \mathbf{x}_{J^{\prime}}}=0$ and hence $\mathcal{H}_{J^{\prime}}\left( h_{J^{\prime \prime}}\right) =0$. Thus, $\left( \sum_{J^{\prime}\subset J}\mathcal{H}_{J^{\prime}}\right) \left( h_{J^{\prime \prime}}\right) =\left( \sum_{J^{\prime}\subseteq J^{\prime \prime}}\mathcal{H}_{J^{\prime}}\right) \left( h_{J^{\prime \prime}}\right) =h_{J^{\prime \prime}}$ and
\begin{align*}
\left( \sum_{J^{\prime}\subset J}\mathcal{H}_{J^{\prime}}\right) \left( h\right)
& =\left( \sum_{J^{\prime}\subset J}\mathcal{H}_{J^{\prime}}\right) \left( \sum_{J^{\prime \prime}\subset J} h_{J^{\prime \prime}}\right) 
=\sum_{J^{\prime \prime}\subset J}\left( \sum_{J^{\prime}\subset J}\mathcal{H}_{J^{\prime}}\right) \left( h_{J^{\prime \prime}}\right)\\
& =\sum_{J^{\prime \prime}\subset J} h_{J^{\prime \prime}}=h.
\end{align*}
\end{proof}

\begin{enumerate}
\renewcommand{\labelenumi}{(\arabic{enumi}) }
    \item From the definition of $\mathcal{H}_\emptyset$ and noting that $\mathcal{H}_J = (I - \mathbb{E}) \circ \mathcal{K}_J$ for $J \neq \emptyset$, we observe that (P1) is always satisfied.

    \item The proof of  (P2)* and (P3)* $\Rightarrow$ (P2) and (P3) can be obtained by induction.

    First, when $J=\{ j\}$, assume that $\frac{\partial g}{\partial x_j}=0$. From (P2)*, it follows that $\frac{\partial \mathcal{L}_j(g)}{\partial x_j}=0$, implying that $\mathcal{L}_j\left( g\right) =\mathrm{const}$. Consequently, $\mathcal{H}_j\left( g\right) =\left( I-\mathbb{E}\right) \circ \mathcal{L}_j\left( g\right) =0$, so (P2) is satisfied. For $g\in V_j$, $\frac{\partial}{\partial x_j}\left( \mathcal{L}_j\left( g\right) -g\right) =0$ by (P3)*, leading to $\mathcal{L}_j\left( g\right) -g=\mathrm{const}$ and hence to $0=\left( I-\mathbb{E}\right) \circ \left( \mathcal{L}_j\left( g\right) -g\right) =\mathcal{H}_j\left( g\right) -g+\mathbb{E}(g)$. Therefore, (P3) is satisfied.

    Now consider the case $\left| J\right| >1$ and assume that for any $J^{\prime}\subset J$, $\mathcal{H}_{J^{\prime}}\left( g\right) =0$ if $\frac{\partial g}{\partial \mathbf{x}_{J^{\prime}}}=0$ (corresponding to (P2)) and $\left( \sum_{J^{\prime \prime}\subseteq J^{\prime}}\mathcal{H}_{J^{\prime \prime}}\right) \left( g\right) =g$ if $g\in V_{J^{\prime}}$ (corresponding to(P3)).    
    If $\frac{\partial g}{\partial \mathbf{x}_J}=0$, (P2)* implies that $\frac{\partial \mathcal{L}_J\left( g\right) }{\partial \mathbf{x}_J}=0$. From the lemma, it follows that $\left( \sum_{J^{\prime}\subset J}\mathcal{H}_{J^{\prime}}\right) \circ \mathcal{L}_J(g)=\mathcal{L}_J(g)$. Therefore, $\mathcal{H}_J(g)=\left( I-\sum_{J^{\prime}\subset J}\mathcal{H}_{J^{\prime}}\right) \circ \mathcal{L}_J(g)=\mathcal{L}_J\left( g\right) -\sum_{J^{\prime}\subset J}\mathcal{H}_{J^{\prime}}\circ \mathcal{L}_J(g)=0$, satisfying (P2). For $g\in V_J$, (P3)* implies that $\frac{\partial}{\partial \mathbf{x}_J}\left( \mathcal{L}_J\left( g\right) -g\right) =0$. Defining $h:=\mathcal{L}_J\left( g\right) -g$, it follows from the lemma that $\left( \sum_{J^{\prime}\subset J}\mathcal{H}_{J^{\prime}}\right) \left(h\right) =h$. Therefore, we have
    \begin{align*}
        \sum_{J^{\prime}\subseteq J}\mathcal{H}_{J^{\prime}}\left( g\right) & =\sum_{J^{\prime}\subset J}\mathcal{H}_{J^{\prime}}\left( g\right) +\mathcal{H}_J\left( g\right)\\
        & =\sum_{J^{\prime}\subset J}\mathcal{H}_{J^{\prime}}\left( g\right) +\left( I-\sum_{J^{\prime}\subset J}\mathcal{H}_{J^{\prime}}\right) \circ \mathcal{L}_J\left( g\right) \\
        & =\sum_{J^{\prime}\subset J}\mathcal{H}_{J^{\prime}}\left( g\right) +\left( I-\sum_{J^{\prime}\subset J}\mathcal{H}_{J^{\prime}}\right) \left( g+h\right)\\
& =g+\left( I-\sum_{J^{\prime}\subset J}\mathcal{H}_{J^{\prime}}\right) (h)=g,
    \end{align*}
so (P3) is satisfied.

    \item To prove that (P2) $\Rightarrow$ (P2)*, we assume that $\frac{\partial g}{\partial \mathbf{x}_J}=0$. From (P2), we have $0=\mathcal{H}_J\left( g\right) =\left( I-\mathbb{E}\right) \circ \mathcal{K}_J\left( g\right)$, from which $\mathcal{K}_J(g)=\mathrm{const}$ follows. Therefore, from the definition of $\mathcal{K}_J$,
    \[\mathrm{const}=\mathcal{K}_J\left( g\right) =\left( I-\sum_{\emptyset\neq J^{\prime}\subset J}\mathcal{K}_{J^{\prime}}\right) \circ \mathcal{L}_J\left( g\right) =\mathcal{L}_J\left( g\right) -\sum_{\emptyset \neq J^{\prime}\subset J}\mathcal{K}_{J^{\prime}}\circ \mathcal{L}_J\left( g\right).\]
Therefore,	
$\mathcal{L}_J\left( g\right) = \sum_{\emptyset \neq J^{\prime}\subset J}\mathcal{K}_{J^{\prime}}\circ \mathcal{L}_J\left( g\right) + \mathrm{const}$.
Because all terms in the summation on the right-hand side belong to $V_{J'}$ with $J' \subset J$, we have $\frac{\partial \mathcal{L}_J\left( g\right)}{\partial \mathbf{x}_J} = 0$.

    \item To prove that (P3) $\Rightarrow$ (P3)*, we assume $g\in V_J$. From (P3), we then have 
    \begin{align*}
        0&=g-\sum_{J^{\prime}\subseteq J}\mathcal{H}_{J^{\prime}}\left( g\right) 
        =\left( I-\mathbb{E}-\sum_{\emptyset \neq J^{\prime}\subseteq J}\left( I-\mathbb{E}\right) \circ \mathcal{K}_{J^{\prime}}\right) \left( g\right)\\
        & =\left( I-\mathbb{E}\right) \circ \left( I-\sum_{\emptyset\neq J^{\prime}\subseteq J}\mathcal{K}_{J^{\prime}}\right) (g).
    \end{align*}
    Therefore, $\left( I-\sum_{\emptyset \neq J^{\prime}\subseteq J}\mathcal{K}_{J^{\prime}}\right) \left( g\right) =g-\sum_{\emptyset \neq J^{\prime}\subseteq J}\mathcal{K}_{J^{\prime}} \left( g\right)=\mathrm{const}$. We then have
    \begin{align*}
    g+\mathrm{const} & =\sum_{\emptyset \neq J^{\prime}\subseteq J}\mathcal{K}_{J^{\prime}}\left( g\right) =\sum_{\emptyset \neq J^{\prime}\subset J}\mathcal{K}_{J^{\prime}}\left( g\right) +\mathcal{K}_J\left( g\right) \\
    & =\sum_{\emptyset \neq J^{\prime}\subset J}\mathcal{K}_{J^{\prime}}\left( g\right) +\left( I-\sum_{\emptyset \neq J^{\prime}\subset J}\mathcal{K}_{J^{\prime}}\right) \circ \mathcal{L}_J\left( g\right)\\
    & =\mathcal{L}_J\left( g\right) +\sum_{\emptyset \neq J^{\prime}\subset J}\mathcal{K}_{J^{\prime}}\circ \left(I-\mathcal{L}_J\left( g\right)\right).
    \end{align*}
    Therefore,	
$\mathcal{L}_J\left( g\right) - g = \sum_{\emptyset \neq J^{\prime}\subset J}\mathcal{K}_{J^{\prime}}\circ \left(I-\mathcal{L}_J\left( g\right)\right) + \mathrm{const}$.
Because all terms in the summation on the right-hand side belong to $V_{J'}$ with $J' \subset J$, we have $\frac{\partial}{\partial \mathbf{x}_J}\left( \mathcal{L}_J\left( g\right) -g\right) =0$.
    \item To prove that (P2) and (P3) $\Rightarrow$ (P4) and (P5), we note that
by definition of $\mathcal{H}_J$, $\mathcal{H}_J\circ \mathcal{H}_J=\mathcal{H}_J\circ \left( I-\sum_{J^{\prime}\subset J}\mathcal{H}_{J^{\prime}}\right) \circ \mathcal{L}_J=\mathcal{H}_J\circ \mathcal{L}_J-\sum_{J^{\prime}\subset J}{\mathcal{H}_J\circ \mathcal{H}_{J^{\prime}}\circ \mathcal{L}_J}$.
From (P2), the second term equals $0$, so the expression simplifies to $\mathcal{H}_J\circ \mathcal{H}_J=\mathcal{H}_J\circ \mathcal{L}_J$. Again using the definition of $\mathcal{H}_J$,
it follows that
$\mathcal{H}_J\circ \mathcal{H}_J=\mathcal{H}_J\circ \mathcal{L}_J=\left( I-\sum_{J^{\prime}\subset J}\mathcal{H}_{J^{\prime}}\right) \circ \mathcal{L}_J\circ \mathcal{L}_J$.
As (P3)* holds by (4) with (P3), $\frac{\partial}{\partial \mathbf{x}_J}\left( \mathcal{L}_J\circ \mathcal{L}_J(g)-\mathcal{L}_J(g)\right) =0$ for any $g$. Letting $h:=\mathcal{L}_J\circ \mathcal{L}_J (g)-\mathcal{L}_J (g)$ and applying the lemma, we have $\left( \sum_{J^{\prime}\subset J}\mathcal{H}_{J^{\prime}}\right) \left( h\right) =h$. Therefore, for any $g$, we have
\begin{align*}
\mathcal{H}_J\circ \mathcal{H}_J\left( g\right)&=\left( I-\sum_{J^{\prime}\subset J}\mathcal{H}_{J^{\prime}}\right) \circ \mathcal{L}_J\circ \mathcal{L}_J\left( g\right)\\
&=\left( I-\sum_{J^{\prime}\subset J}\mathcal{H}_{J^{\prime}}\right) \left( \mathcal{L}_J\left( g\right) +h\right)\\
&=\mathcal{H}_J\left( g\right) +\left( I-\sum_{J^{\prime}\subset J}\mathcal{H}_{J^{\prime}}\right) \left(h\right) =\mathcal{H}_J\left( g\right),
\end{align*}
by which (P4) holds.

If $J^{\prime}\subseteq J$ does not hold, $\frac{\partial \mathcal{H}_J\left( g\right)}{\partial \mathbf{x}_{J^{\prime}}}=0$. Thus, (P2) implies that $\mathcal{H}_{J^{\prime}}\circ \mathcal{H}_J=0$ irrespective of (D). Therefore, to prove that (P5) holds, it is sufficient to prove the case for $J^{\prime}\subset J$. This case can be proven by induction as follows.

First, we note that $\mathcal{H}_\emptyset \circ \mathcal{H}_j=\mathbb{E}\circ (1-\mathbb{E})\circ \mathcal{L}_j=0$. Next, we assume that for any $J^{\prime}\ne J^{\prime \prime}\subset J$, $\mathcal{H}_{J^{\prime}}\circ \mathcal{H}_{J^{\prime \prime}}=0$. Noting that (P4) holds for any $J^{\prime}\subset J$, we then have
\begin{align*}
\mathcal{H}_{J^{\prime}}\circ \mathcal{H}_J&=\mathcal{H}_{J^{\prime}}\circ \left( I-\sum_{J^{\prime \prime}\subset J}\mathcal{H}_{J^{\prime \prime}}\right) \circ \mathcal{L}_J\\
&=\left( \mathcal{H}_{J^{\prime}}-\mathcal{H}_{J^{\prime}}\circ \sum_{J^{\prime \prime}\subset J}\mathcal{H}_{J^{\prime \prime}}\right) \circ \mathcal{L}_J\\
&=\left( \mathcal{H}_{J^{\prime}}-\mathcal{H}_{J^{\prime}}\circ \mathcal{H}_{J^{\prime}}-\mathcal{H}_{J^{\prime}} \circ \sum _{\substack{J^{\prime \prime}\subset J,\\J^{\prime \prime}\neq J^{\prime} }}\mathcal{H}_{J^{\prime \prime}} \right) \circ \mathcal{L}_J\\
&=(\mathcal{H}_{J^{\prime}}-\mathcal{H}_{J^{\prime}}-0)\circ \mathcal{L}_J=0.
\end{align*}
Thus, (P5) also holds.

\item To prove that (P2) and (P3) and (P6)* $\Rightarrow$ (P6), we
assume that the feature variables are independent and apply the lemma with (P2), (P3), and (P6)*. It follows that for any $g\in V_D$, $\left( I-\sum_{J^{\prime}\subset J}\mathcal{H}_{J^{\prime}}\right)\left(\mathcal{L}_J(g) -\mathbb{E}_{\mathrm{X}_{\setminus J }}(g)\right) =0$, implying that $\left( I-\sum_{J^{\prime}\subset J}\mathcal{H}_{J^{\prime}}\right)\circ \mathcal{L}_J =\left( I-\sum_{J^{\prime}\subset J}\mathcal{H}_{J^{\prime}}\right)\circ \mathbb{E}_{\mathrm{X}_{\setminus J }}$. Therefore,
\begin{align*}
\left( \sum_{J^{\prime}\subseteq J}\mathcal{H}_{J^{\prime}}\right) & =\left( \sum_{J^{\prime}\subset J}\mathcal{H}_{J^{\prime}}\right) +\left( I-\sum_{J^{\prime}\subset J}\mathcal{H}_{J^{\prime}}\right) \circ \mathbb{E}_{X_{\setminus J}}\\
& =\mathbb{E}_{X_{\setminus J}}+\sum_{J^{\prime}\subset J}{\mathcal{H}_{J^{\prime}}\circ}\left( I-\mathbb{E}_{X_{\setminus J}}\right) \\
& =\mathbb{E}_{X_{\setminus J}}+\sum_{J^{\prime}\subset J}{\left( I-\sum_{J^{\prime \prime}\subset J^{\prime}}\mathcal{H}_{J^{\prime \prime}}\right) \circ \mathbb{E}_{X_{\setminus J^{\prime}}}\circ \left( I-\mathbb{E}_{X_{\setminus J}}\right)}\\
& =\mathbb{E}_{X_{\setminus J}}+\sum_{J^{\prime}\subset J}{\left( I-\sum_{J^{\prime \prime}\subset J^{\prime}}\mathcal{H}_{J^{\prime \prime}}\right) \circ \left( \mathbb{E}_{X_{\setminus J^{\prime}}}-\mathbb{E}_{X_{\setminus J^{\prime}}}\right)}\\
& =\mathbb{E}_{X_{\setminus J}}.
\end{align*}
Thus, (P6) holds.

\item Finally, to prove that (P6) $\Rightarrow$ (P6)*,
we assume that the feature variables are independent. From (P6), we then have
\begin{align*}
\mathbb{E}_{X_{\setminus J}}\left( g\right)&=\sum_{J^{\prime}\subseteq J}\mathcal{H}_{J^{\prime}}\left( g\right)\\
&=\sum_{J^{\prime}\subset J}\mathcal{H}_{J^{\prime}}\left( g\right) +\mathcal{H}_J\left( g\right)\\
&=\sum_{J^{\prime}\subset J}\mathcal{H}_{J^{\prime}}\left( g\right) +\left( I-\sum_{J^{\prime}\subset J}\mathcal{H}_{J^{\prime}}\right) \circ \mathcal{L}_J\left( g\right).\\
&=\mathcal{L}_J\left( g\right) +\sum_{J^{\prime}\subset J}{\mathcal{H}_{J^{\prime}}\circ \left( I-\mathcal{L}_J\right) \left( g\right)},
\end{align*}
for any $g$. Therefore,	
$\mathcal{L}_J\left( g\right) + \mathbb{E}_{X_{\setminus J}}\left( g\right) = \sum_{J^{\prime}\subset J}{\mathcal{H}_{J^{\prime}}\circ \left( I-\mathcal{L}_J\right) \left( g\right)}$.
Because all terms in the summation on the right-hand side belong to $V_{J'}$ with $J' \subset J$, we have $\frac{\partial}{\partial \mathbf{x}_J}\left( \mathcal{L}_J\left( g\right) -\mathbb{E}_{X_{\setminus J}}\left( g\right) \right) =0$.

 \end{enumerate}

\section{Conclusion}\label{section_conclusion}
This paper mathematically discussed the foundation of methods that decompose prediction functions into their main and interaction effect terms. In this context, we mathematically formalized the requirements that any functional decomposition method must satisfy. We termed a decomposition method that fulfils these requirements an ID and presented a related fundamental theorem. Using the theorem, we conducted several verifications and introduced new methods. Specifically, we confirmed that ALE is an ID whereas the PD-based naive decomposition, which is implicitly used in calculations such as the H-statistic, is not an ID. In addition, we introduced an ID termed the PD-based proper decomposition as well as other concrete examples of alternative functional decomposition methods that meet ID requirements.

It is hoped that the mathematical foundation presented in this paper will initiate further studies on specific new methods and research that will broaden our understanding of interaction effects.

\section*{Acknowledgments}
We express our gratitude to Ryo Jumonji for his valuable insights into the details of some propositions in the early stages of this paper.

\bibliographystyle{apacite}
\bibliography{IDpaper}

\end{document}